\documentclass[letterpaper, 10 pt, conference]{ieeeconf}  % Comment this line out
\IEEEoverridecommandlockouts                              % This command is only
\usepackage{graphics} % for pdf, bitmapped graphics files
\usepackage{epsfig} % for postscript graphics files
\usepackage{amsmath} % assumes amsmath package installed
\usepackage{amssymb}  % assumes amsmath package installed
\usepackage{amsthm}
\usepackage{dsfont}
\usepackage{xcolor}
\usepackage{subfigure}
\usepackage{caption}
\usepackage{subcaption}
\usepackage{romannum}
\usepackage{balance}

\newtheorem{theorem}{Theorem}
\newtheorem{lemma}{Lemma}

\newtheorem{remark}{Remark}

\newtheorem{corollary}{Corollary}
\newtheorem{definition}{Definition}

\usepackage{color}

\title{\LARGE \bf
Fundamental limitations of sensitivity metrics for anomaly impact analysis in LTI systems
}

\author{Jingwei Dong, Kangkang Zhang, Anh Tung Nguyen, and André M. H. Teixeira% <-this % stops a space
\thanks{This work is supported by the Swedish Research Council under the grant 2021-06316 and by the Swedish Foundation for Strategic Research.}% <-this % stops a space
\thanks{Jingwei Dong, Anh Tung Nguyen, and André M. H. Teixeira are with the Department of Information Technology, Uppsala University, SE-75105 Uppsala, Sweden {\tt\small \{jingwei.dong; anh.tung.nguyen; andre.teixeira\}@it.uu.se}. Kangkang Zhang is with the College of Automation Engineering, Nanjing University of Aeronautics and Astronautics, Nanjing, 211106, Jiangsu, China
        {\tt\small (zhang.kangkang@nuaa.edu.cn)}.}%
}

\begin{document}

\maketitle
\thispagestyle{empty}
\pagestyle{empty}

%%%%%%%%%%%%%%%%%%%%%%%%%%%%%%%%%%%%%%%%%%%%%%%%%%%%%%%%%%%%%%%%%%%%%%%%%%%%%%%%
% \begin{abstract}
% This paper studies the fundamental limitations of anomaly detection metrics in linear time-invariant (LTI) systems.
% Building on the output-to-output gain (OOG) developed for stealthy injection attacks, we first propose a novel \mbox{input-to-input} gain (IIG) for robust fault detection, which is proved to be less conservative than conventional mixing metrics, including $\mathcal{H}_{\infty}/\mathcal{H}_{\_}$ and $\mathcal{H}_{\infty}/\mathcal{H}_{\infty}$.
% Subsequently, a {\color{red}unified analysis framework} for OOG and IIG is established given their structural similarities.
% Utilizing the Poisson integral relation, we perform a systematic analysis of the fundamental limitations for both OOG and IIG,
% explicitly characterizing their performance bounds imposed by \mbox{non-minimum} phase (NMP) zeros. 
% Finally, a numerical example is employed to validate the results.
% \end{abstract}

\begin{abstract}
This study establishes a connection between the output-to-output gain (OOG), a sensitivity metric quantifying the impact of stealthy attacks, and a novel input-to-input gain (IIG) introduced to evaluate fault sensitivity under disturbances, and investigates their fundamental performance limitations arising from the transmission zeros of the underlying dynamical system.
Inspired by the OOG, which characterizes the maximum performance loss caused by stealthy attacks, the IIG is proposed as a new measure of robust fault sensitivity, and is defined as the maximum energy of undetectable faults for a given disturbance intensity.
Then, using right (for OOG) and left (for IIG) co-prime factorizations, both metrics are expressed as the~$\mathcal{H}_{\infty}$ norm of a ratio of the numerator factors. 
This unified representation facilitates a systematic analysis of their fundamental limitations.
Subsequently, by utilizing the Poisson integral relation, theoretical bounds for the IIG and OOG are derived,
explicitly characterizing their fundamental limitations imposed by system \mbox{non-minimum} phase (NMP) zeros. 
Finally, a numerical example is employed to validate the results.
\end{abstract}

%%%%%%%%%%%%%%%%%%%%%%%%%%%%%%%%%%%%%%%%%%%%%%%%%%%%%%%%%%%%%%%%%%%%%%%%%%%%%%%%
\section{INTRODUCTION}
The selection of sensitivity metrics is critical in model-based anomaly detection, including both fault detection~\cite{gao2015survey} and the emerging cyber-attack detection driven by the development of cyber-physical systems~\cite{pasqualetti2013attack}.
Conventional approaches in robust fault detection often rely on mixing metrics such as~$\mathcal{H}_{\infty}/\mathcal{H}_{\infty}$ and $\mathcal{H}_{\infty}/\mathcal{H}_{\_}$ to trade off disturbance rejection against fault sensitivity~\cite{ding2008model,dong2025robust}.
However, these metrics are conservative as they consider disturbance robustness and best-case (or worst-case) fault sensitivity separately. 
Additionally, while these metrics can be adapted to deal with security problems, they are ill-posed for analyzing deliberately crafted attacks, e.g., stealthy attacks \cite{teixeira2015secure,Zhang2022}, which evade conventional detection methods.

To address the limitations of conventional metrics, the authors in~\cite{AT2015Strategic} introduced an OOG metric to analyze the impact of stealthy attacks.
%In particular, it quantifies a ratio of the impacts of one attack on two outputs: performance outputs and detection outputs.
In particular, it quantifies the attack impact by calculating the maximum loss in the performance output caused by stealthy attacks, with the stealthiness described by the residual output.
Unlike mixing metrics mentioned above, attack destruction and detectability are jointly integrated in the OOG.
Building on this concept and its advantages, we extend OOG to robust fault detection and develop a novel IIG metric.
The IIG provides an input-based counterpart that concurrently characterizes fault sensitivity and disturbance robustness.
%The IIG quantifies a ratio of the impacts of two inputs: faults and disturbances, on detection outputs, providing an input-based counterpart that concurrently characterizes fault sensitivity and disturbance robustness.

Notably, all aforementioned sensitivity metrics involve conflicting objectives, e.g., maximizing fault sensitivity while minimizing disturbance effects through~$\mathcal{H}_{\infty}/\mathcal{H}_{\_}$. 
The conflicting nature and some intrinsic system properties (such as NMP zeros, unstable poles) impose fundamental limitations.
To the best of our knowledge, similar limitations for these metrics remain unexplored in the literature.
Fortunately, insights from fundamental limitations in \mbox{closed-loop} feedback control systems offer valuable reference.
The study traces back to Bode~\cite{bode1945network}, who introduced the renowned \mbox{gain-phase} relation and sensitivity integral formula for \mbox{single-input single-output} systems. 
Building on this foundation, Freudenberg and Looze~\cite{freudenberg1985right} extended the analysis through the Poisson integral relation, demonstrating how NMP zeros impose fundamental limitations on the achievable performance of sensitivity and complementary sensitivity functions. 
A review of fundamental limitations in feedback control systems is available in~\cite{chen2019fundamental}.

Based on the above discussion, this paper investigates the performance limitations of both OOG and IIG on linear time-invariant (LTI) systems.
The contributions are summarized as follows.
 (1) \textit{Novel sensitivity metric}: 
        %Inspired by OOG between two outputs, we propose IIG between two nputs: faults and disturbances, for robust fault detection.
        Inspired by OOG, we propose IIG for robust fault detection from the perspective of input signals, which allows for joint consideration of fault sensitivity and disturbance robustness.
 (2) \textit{Unified analysis method}: For systems with one-dimensional input and output signals, we propose a unified analysis for OOG and IIG given their analogy (Theorem~\ref{thm: unified expression}). 
 (3) \textit{Fundamental limitations}: 
    Utilizing the Poisson integral relation, we establish theoretical bounds for both OOG and IIG (Theorem~\ref{thm: performance}), revealing their fundamental limitations relating to NMP zeros.

% The outline of the paper is organized as follows. Section~\ref{sec: problem statement} states the problem to be studied in this paper. 
% Section~\ref{sec: main results} presents the main results. 
% In Section~\ref{sec: numerical example}, the developed results are validated through a numerical example. 
% Finally, Section~\ref{sec: conclusions} concludes the paper with future directions. 

\noindent
\textbf{Notation:} 
Sets~$\mathbb{R}$, $\mathbb{R}_+$, and~$\mathbb{R}^n$ denote real numbers, positive real numbers, and the space of~$n$-dimensional \mbox{real-valued} vectors, respectively. 
Given a signal $u = \{u(t)\}_{t \in \mathbb{R}_+}$ and a transfer function~$\mathds{T}$, $\mathds{T}[u]$ denotes the output of the system in response to~$u$ with zero initial conditions.
Considering the time-horizon $[0,T]$ where $T\in\mathbb{R}_+$, the $\mathcal{L}_2$ norm of $u$ over~$[0,T]$ is denoted as~$\|u\|^2_{[0,T]}$.
Define the space of square integrable signals as~$\mathcal{L}_2 \triangleq \{u:\|u\|^2_{[0,\infty]}<\infty\}$ and the extended space~$\mathcal{L}_{2e} \triangleq \{u:\|u\|^2_{[0,T]}<\infty, ~\forall T\in\mathbb{R}_+\}$.
For a complex number $s$, $\bar{s}$ and Re$(s)$ denote its complex conjugate and real part, respectively.
The set of NMP zeros of~$\mathds{T}$ is denoted as~$\mathcal{Z}_{\mathds{T}}\triangleq \{s: \mathds{T}(s)=0, ~\textup{Re}(s) >0 \}$.

%%%%%%%%%%%%%%%%%%%%%%%%%%%%%%%%%%%%%%%%%%%%%%%%%%%%%%%%%%%%%%%%%%%%%%%%%%%%%%%%
\section{Problem formulation}\label{sec: problem statement}
In this section, we first introduce the concepts of OOG and IIG. Then, we formalize the problem of fundamental limitation analysis studied in this paper. 

\subsection{Output-to-output gain}
Consider the following continuous-time LTI system 
\begin{align}\label{eq: LTI-SS}
\Sigma:\left\{ \begin{array}{l}
     \dot{x}_1(t) = Ax_1(t) + B a(t)  \\
     y_r(t) = C_r x_1(t) + D_r a(t)\\
     y_p(t) = C_p x_1(t) + D_p a(t),
\end{array}
\right.
\end{align}
where $x_1(t) \in \mathbb{R}^{n_{x_1}}$,  $y_r(t) \in \mathbb{R}^{n_{y_r}}$, $y_p(t) \in \mathbb{R}^{n_{y_p}}$, and $a(t) \in \mathbb{R}^{n_a}$ denote the state, residual output, performance output, and attack signal, respectively. 
%The signal $a(t) \in \mathbb{R}^{n_a}$ is the attack signal. 
Matrices $A,B,C_r,D_r,C_p$, and $D_p$ are known with appropriate dimensions. 
The system~\eqref{eq: LTI-SS} is denoted as $\Sigma \triangleq (A,B, [C_r^\top,C_p^\top]^\top, [D_r^\top, D_p^\top]^\top)$
and the pair $(A,B)$ is assumed to be controllable.
The transfer functions from $a$ to $y_r$ and from $a$ to $y_p$ are denoted as~$\mathds{T}_{a y_r}$ and~$\mathds{T}_{a y_p}$, respectively. 

The system~\eqref{eq: LTI-SS} represents an 
augmented version of the \mbox{closed-loop} model illustrated in Fig.~\ref{fig: OOG}, which integrates the physical plant, feedback controller, and anomaly detector~\cite[Chapter 6]{ferrari2021safety}.
The performance output $y_p$ is used to evaluate the closed-loop performance loss caused by attack signals.
The residual $y_r$ can be viewed as the output of the anomaly detector, which is used to detect anomalies in the system. 
The alarm is triggered when $y_r$ deviates significantly from zero, i.e., $\|y_r\|^2_{[0,T]} \geq \epsilon$ where~$\epsilon \in \mathbb{R}_+$ is a given threshold, indicating the presence of anomalies. Without loss of generality, we set~$\epsilon=1$ throughout this paper.

Consider system \eqref{eq: LTI-SS}, a stealthy adversary aims to inject a crafted attack signal~$a$ to disrupt the system's behavior while remaining stealthy to the anomaly detector.
In the context of the finite time interval, the level of attack disruption can be evaluated through the energy increase of the performance output,~i.e.,~$\|y_p\|^2_{[0,T]}$.
To remain stealthy, $\|y_r\|^2_{[0,T]}$ should be below the threshold, i.e., $\|y_r\|^2_{[0,T]} \leq 1$. 
Given the description of the system and the adversary model, and letting $T$ go to infinity, the stealthy attack~$a$ can be characterized through the following non-convex optimization problem~\cite{AT2015Strategic}:
\begin{align}\label{eq: OOG}
   \|\Sigma\|^2_{\mathcal{L}_{2e},y_p \leftarrow y_r} \triangleq& \sup_{a \in \mathcal{L}_{2e}} ~\|y_p\|^2_{\mathcal{L}_2} \notag\\
   &~~~
   \text{s.t.} ~\eqref{eq: LTI-SS},~x_1(0)=0, ~\|y_r\|^2_{\mathcal{L}_2} \leq 1.
\end{align} 
Here, $\|\Sigma\|^2_{\mathcal{L}_{2e},y_p \leftarrow y_r}$ is defined as the OOG of $\Sigma$, which characterizes the maximum disruption on the performance output induced by the stealthy attack~$a$. 
\begin{figure}
    \centering
    \includegraphics[width=0.8\linewidth]{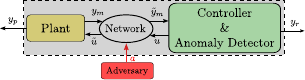}
    \caption{\small Structure of networked control systems under attacks.}
    \label{fig: OOG}
    \vspace{-0.6cm}
\end{figure}

\begin{figure}
    \centering
    \includegraphics[width=0.6\linewidth]{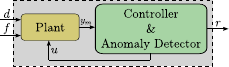}
    \caption{\small Structure of systems under faults and disturbances.}
    \label{fig: IIG}
    \vspace{-0.6cm}
\end{figure}

% \begin{figure}
%     \centering
%     \subfigure[]{\includegraphics[width=0.35\textwidth]{Fig_OOG.pdf}} 
%     \subfigure[]{\includegraphics[width=0.25\textwidth]{Fig_IIG.pdf}} 
%     \caption{(a) blah (b) blah (c) blah (d) blah}
%     \label{fig:foobar}
% \end{figure}

\subsection{Input-to-input gain}
The aforementioned OOG focuses on evaluating the disruption level and detectability of stealthy attacks by jointly examining the effects of an attack $a$ on two outputs $y_p$ and $y_r$.
This concept is dual to robust fault detection~\cite{ding2008model}, which analyzes the effects of two inputs (faults and disturbances) on the same output.
However, a key distinction is that the effects of the two inputs are analyzed separately in trivial robust fault detection (e.g., by setting the other to zero).
Building on the joint analysis idea of OOG, we develop IIG to evaluate the effects of faults and disturbances concurrently.
In what follows, we provide a detailed explanation of this extension.

Let us consider the following continuous-time LTI system
\begin{align}\label{eq: LTI-SS Fault}
\tilde{\Sigma}:\left\{ \begin{array}{l}
     \dot{x}_2(t) = \tilde{A}x_2(t)  + B_d d(t) + B_f f(t)  \\
     r(t)   = C x_2(t) + D_d d(t) +  D_f f(t),
\end{array}
\right. 
\end{align}
where $x_2(t) \in \mathbb{R}^{n_{x_2}}$, $d(t) \in \mathbb{R}^{n_d}$, $f(t) \in \mathbb{R}^{n_f}$, and~$r(t) \in \mathbb{R}^{n_r}$ denote the state, disturbance, fault, and residual, respectively.
Matrices $\tilde{A},~B_d$, $B_f$, $C$, $D_d$, and $D_f$ are known with appropriate dimensions.
The system~\eqref{eq: LTI-SS Fault} is denoted as $\tilde{\Sigma} \triangleq (\tilde{A},[B_d,B_f],C,[D_d,D_f])$, which can be interpreted as a general form of a closed-loop system incorporating a feedback controller and anomaly detector~\cite{ding2008model} (see Fig.~\ref{fig: IIG}). 
We assume that the pair $(\tilde{A},C)$ is observable, and the disturbance $d$ is bounded, i.e.,~$\|d\|^2_{\mathcal{L}_2} \leq 1$, with the upper bound set to~$1$ without loss of generality.

The residual~$r$ generated by the anomaly detector is used to indicate the occurrence of faults, which can be written as
\begin{align}\label{eq: r}
r = \mathds{T}_{dr}[d] + \mathds{T}_{fr}[f],
\end{align}
where $\mathds{T}_{dr}$ and $\mathds{T}_{fr}$ are the transfer functions from $d$ to $r$ and from $f$ to $r$, respectively.
It is worth emphasizing that we consider non-decoupled disturbances in~\eqref{eq: LTI-SS Fault}. Otherwise, one can use disturbance decoupling techniques as in~\cite{ding2008model}, rendering ~$\mathds{T}_{dr}[d]=0$. In this case, any faults can be detected theoretically. 
In the case of non-decoupled disturbances, the effect of~$f$ on $r$ can be masked by that of~$d$, resulting in undetectable faults, as defined below.

\begin{definition}[Undetectable faults]
    Given a set of disturbances~$\mathcal{D}$, a fault~$f$ is~$\mathcal{D}$-undetectable if there exists a $d \in \mathcal{D}$ such that
\begin{align}\label{eq: undet con}
    \| r \|^2_{\mathcal{L}_2} = \|\mathds{T}_{dr}[d] + \mathds{T}_{fr}[f] \|^2_{\mathcal{L}_2} \leq \|\mathds{T}_{dr}[d]\|^2_{\mathcal{L}_2}.
\end{align}
\end{definition}
The condition \eqref{eq: undet con} defines the undetectability by jointly considering the effects of faults and disturbances. Additionally, $\mathcal{D} = \{0\}$ corresponds to the generic definition of undetectable faults found in~\cite{ding2008model}, i.e.,~$\mathds{T}_{fr}[f] = 0$.

Note that a \textit{larger} undetectable fault in energy for a given disturbance intensity indicates \textit{lower} fault sensitivity. 
Thus, the fault sensitivity of~\eqref{eq: LTI-SS Fault} can be quantified by the maximum energy of undetectable faults satisfying~\eqref{eq: undet con}.
This leads to the following optimization problem:
\begin{align}\label{eq: IIG}
    \|\tilde{\Sigma}\|^2_{\mathcal{L}_{2e},f \leftarrow d} \triangleq&\sup_{f,\,d\in \mathcal{L}_{2e}} ~\|f\|^2_{\mathcal{L}_2} \notag \\
    \text{s.t.} ~&\eqref{eq: LTI-SS Fault}, ~x_2(0)=0, ~ \|d\|^2_{\mathcal{L}_2} \leq  1, \\
    & \|\mathds{T}_{dr}[d] + \mathds{T}_{fr}[f] \|^2_{\mathcal{L}_2} \leq \|\mathds{T}_{dr}[d]\|^2_{\mathcal{L}_2}. \notag                 
\end{align}
The optimal value $\|\tilde{\Sigma}\|^2_{\mathcal{L}_{2e},f \leftarrow d}$ serves as the novel metric for the fault sensitivity under a certain disturbance intensity and is defined as the IIG of $\tilde{\Sigma}$.

\subsection{Problem statement}
The OOG regarding the stealthiness of an attack and the IIG characterizing the detectability of a fault are both formulated as the maximization problem in~\eqref{eq: OOG} and~\eqref{eq: IIG}, respectively. In fact, a similarity connection between OOG and IIG is hidden behind, which can be observed from the forms of~\eqref{eq: OOG} and~\eqref{eq: IIG}, and needs to be clarified. In addition, the maximums in~\eqref{eq: OOG} and~\eqref{eq: IIG} actually have their own limits that cannot be broken due to inherent system properties, such as NMP zeros and unstable poles.
%It can be observed from the two optimization problems~\eqref{eq: OOG} and~\eqref{eq: IIG} that, both of them seek to balance two conflicting design objectives to achieve optimal performance. However, due to inherent system properties, such as NMP zeros and unstable poles, optimal performance has its limits and may not even exist in some cases.
%Therefore, this paper investigates the fundamental limitations of OOG and IIG.
%{\color{red}To simplify the subsequent analysis, this study considers one-dimensional input and output signals in~\eqref{eq: LTI-SS} and~\eqref{eq: LTI-SS Fault}, primarily focusing on the following two aspects:}
Taking the above two issues into consideration, this paper focuses on the following problems:
\begin{enumerate}
    \item  Reveal the hidden connection between OOG and IIG by developing a unified expression;%Develop a unified analysis framework for OOG and IIG given their structural similarities;
    \item Based on this unified expression, the fundamental limitations of both OOG and IIG are derived using the Poisson integral relation. %Under the unified framework, identify the fundamental limitations of both OOG and IIG. 
\end{enumerate}
To simplify the subsequent analysis, this study considers one-dimensional input and output signals in~\eqref{eq: LTI-SS} and~\eqref{eq: LTI-SS Fault}, and provides some new insights into the fault diagnosis and security analysis. The OOG and IIG for more general multi-input-multi-output systems will be studied in future work. %By studying systems with one-dimensional input and output signals, we provide some new insights into the fault diagnosis and security analysis of more general systems, which will be studied in future work.

%%%%%%%%%%%%%%%%%%%%%%%%%%%%%%%%%%%%%%%%%%%%%%%%%%%%%%%%%%%%%%%%%%%%%%%%%%%%%%%%
\section{Main results}\label{sec: main results}

\subsection{Unified analysis method for OOG and IIG}
To the end, some preparatory results are introduced.
The condition~\eqref{eq: undet con} is first reformulated in the following lemma.

\begin{lemma}[Reformulated undetectability condition]\label{lem: reform con}
    Consider the LTI system~\eqref{eq: LTI-SS Fault} with the disturbance~$\|d\|^2_{\mathcal{L}_2} \leq  1$. 
    The undetectability condition~\eqref{eq: undet con} can be guaranteed if the following conditions hold: 
    \begin{align}\label{eq: New_Undet_Con}
         \mathds{T}_{fr}[f] = \mathds{T}_{dr}[\tilde{d}],  ~\|\tilde{d}\|^2_{\mathcal{L}_2} \leq \xi^2,~\xi \in [0,2],
    \end{align} 
where~$\tilde{d} = -\xi d$.
\end{lemma}

\begin{proof}
    First,~\eqref{eq: undet con} is satisfied when~$-\xi\mathds{T}_{dr}[d] = \mathds{T}_{fr}[f]$ with~$\xi \in [0,2]$ because
    \begin{align*}
        \|\mathds{T}_{dr}[d] + \mathds{T}_{fr}[f] \|^2_{\mathcal{L}_2}  \leq |1-\xi|^2 \|\mathds{T}_{dr}[d]\|^2_{\mathcal{L}_2} \leq \|\mathds{T}_{dr}[d]\|^2_{\mathcal{L}_2}.
    \end{align*}
    Replacing~$-\xi d$ with $\tilde{d}$ in $-\xi\mathds{T}_{dr}[d] = \mathds{T}_{fr}[f]$ yields the first condition in~\eqref{eq: New_Undet_Con}. 
    The inequality~$\|\tilde{d}\|^2_{\mathcal{L}_2} \leq \xi^2$ follows directly from $\|d\|^2_{\mathcal{L}_2} \leq  1$ and $\tilde{d} = -\xi d$.  
    This completes the proof.
\end{proof}

With the stricter condition~\eqref{eq: New_Undet_Con} and setting~$\xi=2$, 
the optimization problem~\eqref{eq: IIG} becomes
\begin{align}\label{eq: Reform_IIG}
    \|\tilde{\underline{\Sigma}}\|^2_{\mathcal{L}_{2e},f \leftarrow \tilde{d}} \overset{\Delta}{=}\sup_{f,\tilde{d}\in \mathcal{L}_{2e}} ~&\|f\|^2_{\mathcal{L}_2} \notag \\
    \text{s.t.} ~~~&\eqref{eq: LTI-SS Fault}, ~x_2(0)=0, \notag\\
    & \mathds{T}_{fr}[f] = \mathds{T}_{dr}[\tilde{d}],
    ~\|\tilde{d}\|^2_{\mathcal{L}_2} \leq 4,
\end{align}
where~$\|\tilde{\underline{\Sigma}}\|^2_{\mathcal{L}_{2e},f \leftarrow \tilde{d}}$ serves as a lower bound for~$\|\tilde{\Sigma}\|^2_{\mathcal{L}_{2e},f \leftarrow d}$ since~\eqref{eq: New_Undet_Con} is a particular instance of~\eqref{eq: undet con}.  
Also, it is straightforward to show by contradiction that the optimal solution to~\eqref{eq: Reform_IIG} is achieved at~$\xi=2$.
While the condition~\eqref{eq: New_Undet_Con} may be restrictive, requiring $d$ and $f$ to have the same output (up to a scale), we note that the condition naturally arises in the context of indistinguishability between faults and disturbances~\cite[Chapter 4]{ding2008model}.
A similar condition was exploited in~\cite{sandberg2016control} to develop security indices with respect to faults masked by disturbances as well.

Next, we recall a fundamental result from dissipative systems theory.
For more details, we refer readers to~\cite[Theorem 4.5]{trentelman1991dissipation} and~\cite{teixeira2019optimal}.

\begin{lemma}[Frequency domain inequality~{\cite{teixeira2019optimal}}]\label{lem: FDI}
    Consider the LTI system in~\eqref{eq: LTI-SS} with~$(A,B)$ controllable and the inequality:
    \begin{align}\label{eq: OOG_ineq}
         \| y_r \|^2_{\mathcal{L}_2} - \| y_p \|^2_{\mathcal{L}_2} \geq 0.
    \end{align}
    (1) For all trajectories of system~\eqref{eq: LTI-SS} with $x_1(0)=0$, a \textit{necessary} condition for~\eqref{eq: OOG_ineq} to hold is 
    \begin{align*}
            \mathds{T}^{\top}_{a y_r}(\bar{s}) \mathds{T}_{a y_r}(s) \geq \mathds{T}^{\top}_{a y_p} (\bar{s})\mathds{T}_{a y_p}(s), \forall s \notin \lambda_A, \textup{Re}(s) \geq 0.
    \end{align*}
    (2) For system~\eqref{eq: LTI-SS} subject to $T-$periodic trajectories, i.e., $x_1(t+T)=x_1(t)$ where $T \in \mathbb{R}_+$, a \textit{necessary and sufficient} condition for~\eqref{eq: OOG_ineq} to hold is
        \begin{align*}
            \mathds{T}^{\top}_{a y_r}(\bar{s}) \mathds{T}_{a y_r}(s) \geq \mathds{T}^{\top}_{a y_p} (\bar{s})\mathds{T}_{a y_p}(s), \forall s \notin \lambda_A, \textup{Re}(s) = 0.
        \end{align*}
        Here, $\lambda_A$ stands for a set of all the eigenvalues of matrix $A$.
\end{lemma}

By performing the right co-prime factorization on~$\mathds{T}_{a y_r}$ and~$\mathds{T}_{a y_p}$, we obtain
\begin{align}\label{eq: RCPF}
    \mathds{T}_{a y_r} =  N_r M_O^{-1}, \quad \mathds{T}_{a y_p} =  N_p M_O^{-1}.
 \end{align}
The left co-prime factorization of~$\mathds{T}_{fr}$ and $\mathds{T}_{dr}$ are given by
\begin{align}\label{eq: LCPF}
    \mathds{T}_{fr} = M^{-1}_{I}N_f, \quad \mathds{T}_{dr} = M^{-1}_{I}N_d.
\end{align}
The detailed co-prime factorization approach can be found in~\cite[Chapter 4]{zhou1996robust}. 
Note that zeros of~$\mathds{T}_{a y_r}$ and~$\mathds{T}_{fr}$ coincide with those of~$N_r$ and~$N_f$, respectively. %We assume that~$\mathds{T}_{a y_r}$ and~$\mathds{T}_{fr}$ contain no NMP zeros.% in the following theorem.

To simplify notations, we define $\gamma^* \triangleq \|\Sigma\|^2_{\mathcal{L}_{2e},y_p \leftarrow y_r}$ and $\tilde{\gamma}^* \triangleq \|\tilde{\underline{\Sigma}}\|^2_{\mathcal{L}_{2e},f \leftarrow \tilde{d}}$.
Now, we are in the position to present the unified analysis for OOG and IIG in the following theorem.

\begin{theorem}[Unified analysis for OOG and IIG]\label{thm: unified expression}
    Consider the continuous-time LTI systems in~\eqref{eq: LTI-SS} and~\eqref{eq: LTI-SS Fault} with one-dimensional input and output signals, and given the co-prime factorization in~\eqref{eq: RCPF} and~\eqref{eq: LCPF} with stable $N_r$ and $N_f$, the OOG in~\eqref{eq: OOG} and the reformulated IIG in~\eqref{eq: Reform_IIG} satisfy:
    \begin{align}\label{eq: unified expression}
         \gamma^* \geq \left\| \frac{N_p}{N_r} \right\|^2_{\mathcal{H}_\infty} ~\text{and} 
         ~\tilde{\gamma}^* \geq 4 \left\| \frac{N_d}{N_f} \right\|^2_{\mathcal{H}_\infty}.
    \end{align}
    Furthermore, the equalities in~\eqref{eq: unified expression} can be achieved if the systems are subject to periodic trajectories.
\end{theorem}
\begin{proof}
    The proof is relegated to Appendix.
\end{proof}

\begin{remark}[NMP zeros in~$\mathds{T}_{a y_r}$ and $\mathds{T}_{fr}$]
   The presence of NMP zeros in~$\mathds{T}_{a y_r}$ (or $\mathds{T}_{fr}$) leads to two cases. First, if $\mathds{T}_{a y_r}$ (or~$\mathds{T}_{fr}$) has at least one different NMP zero with~$\mathds{T}_{a y_p}$ (or~$\mathds{T}_{d r}$, respectively),
    the OOG (or IIG) value becomes infinite. This occurs because certain inputs aligned with these NMP zeros can yield zero output, as demonstrated in~\cite[Theorem 2]{AT2015Strategic}.
    Second, if all NMP zeros of~$\mathds{T}_{a y_r}$ (or~$\mathds{T}_{fr}$) are shared with~$\mathds{T}_{a y_p}$ (or~$\mathds{T}_{d r}$), they cancel out and the problem degrades to the case where~$\mathds{T}_{a y_r}$ (or $\mathds{T}_{fr}$) has no NMP zeros.
\end{remark}

Additionally, we establish relations between the developed IIG metric with the classical $\mathcal{H}_{\infty}/\mathcal{H}_{\infty}$ and $\mathcal{H}_{\infty}/\mathcal{H}_{\_}$ fault detection metrics in the following corollary.
A similar result for OOG can be obtained as well, and is omitted here.% to avoid repetition. 
\begin{corollary}[Relation with existing metrics]\label{cor: up_low_bound}
The bound of the IIG metric developed in~\eqref{eq: unified expression} satisfies: 
\begin{align}\label{eq: Up_low_bound}
     \frac{\| \mathds{T}_{dr} \|^2_{\mathcal{H}_{\infty}}}{\| \mathds{T}_{fr} \|^2_{\mathcal{H}_{\infty}}} \leq \left\| \frac{N_d}{N_f} \right\|^2_{\mathcal{H}_\infty}	\leq \frac{\| \mathds{T}_{dr} \|^2_{\mathcal{H}_{\infty}}}{\| \mathds{T}_{fr} \|^2_{\mathcal{H}_{\_}}}.
\end{align} 
\end{corollary}
\begin{proof}
By definition, the~$\mathcal{H}_{\infty}$ norm $\|\mathds{T}_{fr}\|_{\mathcal{H}_{\infty}}$ and the~$\mathcal{H}_{\_}$ index $\|\mathds{T}_{fr}\|_{\mathcal{H}_{\_}}$ represent the largest and smallest singular values of $\mathds{T}_{fr}$, respectively. 
Furthermore, the two relations hold: 
$\|\mathds{T}^{-1}_{fr}\|_{\mathcal{H}_{\_}} = \|\mathds{T}_{fr}\|^{-1}_{\mathcal{H}_{\infty}}$ ~\text{and}~ $\|\mathds{T}^{-1}_{fr}\|_{\mathcal{H}_{\infty}} = \|\mathds{T}_{fr}\|^{-1}_{\mathcal{H}_{\_}}$,
leading to
\begin{align*}
    &\frac{\| \mathds{T}_{dr} \|^2_{\mathcal{H}_{\infty}}}{\| \mathds{T}_{fr} \|^2_{\mathcal{H}_{\infty}}}
    =\|\mathds{T}^{-1}_{fr}\|^2_{\mathcal{H}_{\_}}\| \mathds{T}_{dr} \|^2_{\mathcal{H}_{\infty}} 
    \leq \|\mathds{T}^{-1}_{fr} \mathds{T}_{dr}\|^2_{\mathcal{H}_{\infty}}, \\
    &\|\mathds{T}^{-1}_{fr} \mathds{T}_{dr}\|^2_{\mathcal{H}_{\infty}}
    \leq \|\mathds{T}^{-1}_{fr} \|^2_{\mathcal{H}_{\infty}} \|\mathds{T}_{dr}\|^2_{\mathcal{H}_{\infty}} = \frac{\| \mathds{T}_{dr} \|^2_{H_{\infty}}}{\| \mathds{T}_{fr} \|^2_{H_{\_}}}.
\end{align*}
From the left co-prime factorization given in~\eqref{eq: LCPF}, we have $\mathds{T}^{-1}_{fr} \mathds{T}_{dr} = N^{-1}_f N_d$. This completes the proof.  
\end{proof}

\subsection{Fundamental limitation analysis}
From the result presented in Theorem~\ref{thm: unified expression}, both bounds in~\eqref{eq: unified expression} are expressed as the~$\mathcal{H}_{\infty}$ norm of a transfer function.  
This motivates us to analyze their fundamental limitations based on whether they are proper or not.
For the first scenario, where~$N_p/N_r$ and $N_d/N_f$ are non-proper ratios, the result is provided in the following lemma.

\begin{lemma}[Performance limitations in non-proper case]\label{lem: improper ratios} 
If the transfer functions $N_p/N_r$ and $N_d/N_f$ are not proper, the corresponding performance bounds for OOG and IIG metrics are infinite, i.e.,~$\left\| N_p/N_r \right\|^2_{\mathcal{H}_\infty} = \left\| N_d/N_f \right\|^2_{\mathcal{H}_\infty} = \infty$.
\end{lemma}

The result in Lemma~\ref{lem: improper ratios} is straightforward. For an non-proper transfer function~$\mathds{T}$, it is easy to verify that $\|\mathds{T}(j\omega)\|_{\mathcal{H}_\infty} \rightarrow \infty$ as $\omega \rightarrow \infty$.
This implies two key interpretations:
(\Romannum{1}) if $N_p/N_r$ is non-proper for OOG, high-frequency attack signals can induce significant performance degradation while remaining stealthy;
(\Romannum{2}) if $N_d/N_f$ is non-proper for IIG, high-frequency faults can be masked by high-frequency disturbances, rendering faults difficult to detect.

To derive the performance limitations in~\eqref{eq: unified expression} for proper systems, we introduce the following notation. Let us define:
\begin{align}\label{eq: OOG fac}
    S_O %= 
    \triangleq \frac{N_p}{N_r}, \quad
    P_O %= 
    \triangleq 1-\frac{N_p}{N_r}.
\end{align}
Assume $S_O$ and $P_O$ have $n_\alpha$ and $n_\beta$ numbers of NMP zeros, respectively. 
The NMP zeros of~$S_O$ are denoted as~$\alpha_i \in \mathcal{Z}_{S_O}, i=\{1,\dots,n_{\alpha}\}$, and those of~$P_O$ are denoted as $\beta_i \in \mathcal{Z}_{P_O}, i=\{1,\dots,n_{\beta}\}$.
Then, $S_O$ and $P_O$ are factorized as
\begin{align}\label{eq: BlaschkeFac}
    S_O = \tilde{S}_O \mathcal{B}_{S_O}, \quad
    P_O= \tilde{P}_O \mathcal{B}_{P_O},
\end{align}
where $\tilde{S}_O$ and $\tilde{P}_O$ are the corresponding minimum-phase parts. 
The Blaschke products of NMP zeros in $S_O$ and $P_O$ are denoted as $\mathcal{B}_{S_O}$ and $\mathcal{B}_{P_O}$, which are given by
\begin{align*}
    \mathcal{B}_{S_O}(s) = \prod^{n_\alpha}_{i=1} \frac{s-\alpha_i}{s+\bar{\alpha}_i}, \quad
    \mathcal{B}_{P_O}(s) = \prod^{n_\beta}_{i=1} \frac{s-\beta_i}{s+\bar{\beta}_i}.
\end{align*}
Note that $\mathcal{B}_{S_O}$ and $\mathcal{B}_{P_O}$ are all-pass factors, leading to the fact that $|S_O(j\omega)|=|\tilde{S}_O(j\omega)|$ and $|P_O(j\omega)|=|\tilde{P}_O(j\omega)|$. If the set of NMP zeros is empty, we define the corresponding Blaschke product to be $1$.
Similarly, for IIG, we define
\begin{align*}
    S_I \triangleq \frac{N_d}{N_f}, \quad 
    P_I \triangleq 1-\frac{N_d}{N_f},
\end{align*}
with their NMP zeros denoted as~$\mu_{i} \in \mathcal{Z}_{S_I}, i=\{1,\dots,n_{\mu}\}$ and ~$\nu_{i} \in \mathcal{Z}_{P_I}, i=\{1,\dots,n_{\nu}\}$ by assuming $S_I$ and $P_I$ have $n_\mu$ and $n_\nu$ numbers of NMP zeros, respectively.
Factorizing $S_I$ and $P_I$ yields
\begin{align}\label{eq: IIG fac}
    S_I = \tilde{S}_I \mathcal{B}_{S_I}, \quad 
    P_I= \tilde{P}_I \mathcal{B}_{P_I},
\end{align}
where~$\tilde{S}_I$ and $\tilde{P}_I$ are the minimum-phase parts, and the Blaschke products $\mathcal{B}_{S_I}$ and $\mathcal{B}_{P_I}$ are in the form of 
\begin{align}
    \mathcal{B}_{S_I}(s) = \prod^{n_\mu}_{i=1} \frac{s-\mu_i}{s+\bar{\mu}_i}, \quad
    \mathcal{B}_{P_I}(s) = \prod^{n_\nu}_{i=1} \frac{s-\nu_i}{s+\bar{\nu}_i}.
\end{align}
Then, the limitations of OOG and IIG are ready to be presented below.

\begin{theorem}[Performance limitations in proper case]\label{thm: performance}
Consider the expressions for OOG and IIG developed in~\eqref{eq: unified expression} and the factorization in~\eqref{eq: OOG fac} and~\eqref{eq: IIG fac}. 
The performance limitations for proper $S_O$ and $S_I$ satisfy
\begin{align}\label{eq: performance bound}
        & ||S_O||_{\mathcal{H}_{\infty}} \geq \max_{\alpha_h\in\mathcal{Z}_{S_O},\beta_k\in\mathcal{Z}_{P_O}} \left\{ |\mathcal{B}^{-1}_{S_O}(\beta_k)|, |\mathcal{B}^{-1}_{P_O}(\alpha_h)|-1 \right\},  \notag\\
        &  ||S_I||_{\mathcal{H}_{\infty}} \geq \max_{\mu_h\in\mathcal{Z}_{S_I},\nu_k\in\mathcal{Z}_{P_I}} \left\{ |\mathcal{B}^{-1}_{S_I}(\nu_k)|, |\mathcal{B}^{-1}_{P_I}(\mu_h)|-1 \right\}.
\end{align}
\end{theorem}
\begin{proof}
    The proof is relegated to Appendix.
\end{proof}
Theorem~\ref{thm: performance} relates the performance limitations of OOG and IIG to NMP zeros of the system.
For instance, the lower bound of $\|S_O\|_{\mathcal{H}_{\infty}}$ is larger than or equal to~$1$.
Specifically, if at most one of $S_O$ and $P_O$ has NMP zeros, the lower bound of $\|S_O\|_{\mathcal{H}_{\infty}}$ is~$1$, indicating that the undetectable attack satisfying $\|y_p\|^2_{\mathcal{L}_2} \geq 1$ always exists.
The lower bound of $\|S_O\|_{\mathcal{H}_{\infty}}$ is strictly larger than~$1$ if both $S_O$ and $P_O$ have NMP zeros, and is determined by the distance between their NMP zeros. When the NMP zeros~$\alpha_h$ and $\beta_k$ are far apart, the lower bound approaches $1$. Furthermore, the bound increases significantly when the NMP zeros~$\alpha_h$ and $\beta_k$ are close to each other, indicating that stealthy attacks can generate significant effects.
 %indicating a large disruption level with undetectable attacks. 
%Moreover, it is worth pointing out that $S_O$ and $P_O$ cannot share common NMP zeros because such a zero would also be an NMP zero of $N_r$, which, by assumption, has no NMP zeros in this study. 

%%%%%%%%%%%%%%%%%%%%%%%%%%%%%%%%%%%%%%%%%%%%%%%%%%%%%%%%%%%%%%%%%%%%%%%%%%%%%%%%
\section{Numerical examples}\label{sec: numerical example}

This section mainly focuses on validating the obtained fundamental bound for IIG, given that a validation for OOG can be done similarly.
Consider the transfer functions corresponding to system~\eqref{eq: LTI-SS Fault} as follows:
\begin{align*}
    &\mathds{T}_{fr}(s) = \frac{(s+0.1)(s+0.2)(s+0.6)}{(s+0.3)(s+0.4)(s+0.5)}, \\
    &\mathds{T}_{dr}(s) = \frac{(s+1)(s-0.04)(s-\tau)}{(s+0.3)(s+0.4)(s+0.5)},
\end{align*}
where $\mathds{T}_{dr}$ contains a NMP zero at $0.04$ and a zero~$\tau$.
In the following, we examine how the performance bound in~\eqref{eq: performance bound} changes with respect to $\tau$ that varies between $-20$ and $20$.

\begin{figure}
\centering
  \centering
  \includegraphics[scale=.38]{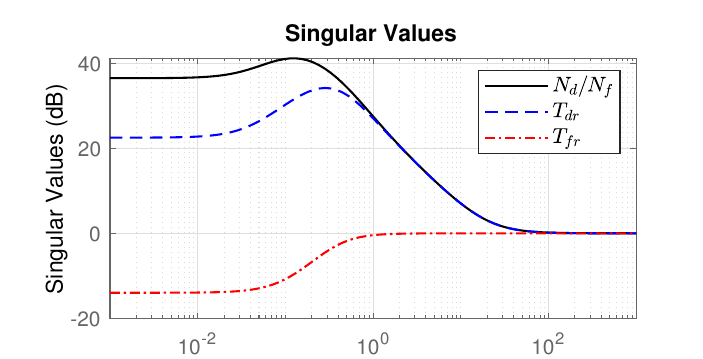}
  \captionof{figure}{\small Singular values of $N_d/N_f$, $\mathds{T}_{dr}$, and $\mathds{T}_{fr}$.}
  \label{fig: FrequencyResponse}
  \vspace{-0.4cm}
\end{figure}

\begin{figure}
\centering
\begin{minipage}{.45\columnwidth}
  \centering
  \includegraphics[scale=.35]{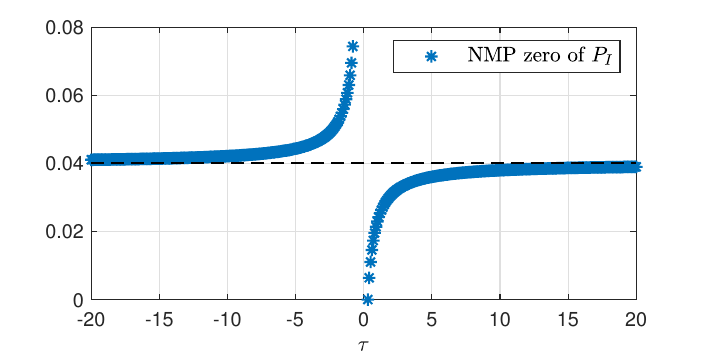}
  \captionof{figure}{\small NMP zero of $P_I$ varying with $\tau$.}
  \label{fig: NMP zero}
\end{minipage}
\hspace{10pt}
\begin{minipage}{.45\columnwidth}
  \centering
  \includegraphics[scale=.35]{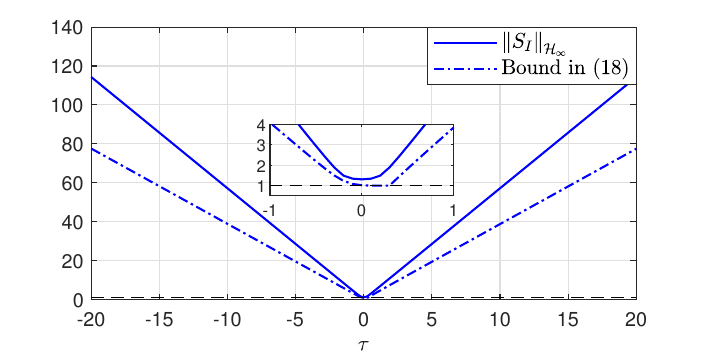}
  \captionof{figure}{\small Performance bound for $\|S_I\|_{\mathcal{H}_{\infty}}$.}
  \label{fig: PerformanceBound}
\end{minipage}
\vspace{-0.7cm}
\end{figure}

The first result is presented in
Fig.~\ref{fig: FrequencyResponse}, where the singular values of $N_d/N_f$, $\mathds{T}_{dr}$, and $\mathds{T}_{fr}$ with $\tau=20$ are shown. 
It can be observed that~$\|\mathds{T}_{dr}\|_{\mathcal{H}_{\infty}} = 57.17$, $\|\mathds{T}_{fr}\|_{\mathcal{H}_{\infty}} = 1$, $\|\mathds{T}_{fr}\|_{\mathcal{H}_{\_}} = 0.2$, and $\|N_d/N_f\|_{\mathcal{H}_\infty} = 114.33$, which indicates 
$\| \mathds{T}_{dr} \|_{\mathcal{H}_{\infty}} / \| \mathds{T}_{fr} \|_{\mathcal{H}_{\infty}} \leq 114.33	\leq\| \mathds{T}_{dr} \|_{\mathcal{H}_{\infty}} / \| \mathds{T}_{fr} \|_{\mathcal{H}_{\_}} $.
This verifies the result in Corollary~\ref{cor: up_low_bound} and also demonstrates that the IIG metric is less conservative.

The simulation results for the performance bound are presented in Figs.~\ref{fig: NMP zero}-\ref{fig: PerformanceBound}. 
As shown in Fig.~\ref{fig: NMP zero}, as $\tau$ increases in both directions, the zero of $N_f-N_d$ approaches that of $N_d$ at $0.04$. 
%In this example, as the magnitude of $\tau$ increases in both directions, the NMP zero of $N_f-N_d$ gets close to that of $N_d$ at $0.04$, as shown in Fig.~\ref{fig: NMP zero}. 
As a result, the lower bound of $\|S_I\|_{\mathcal{H}_{\infty}}$ increases when $\tau$ deviates from $0$, as depicted in Fig.~\ref{fig: PerformanceBound}. This aligns with the theoretical results in Theorem \ref{thm: performance}.
Moreover, as depicted in the zoom-in figure inside Fig.~\ref{fig: PerformanceBound}, when $\tau$ is around zero, the lower bound is~$1$. This suggests that either $N_d$ or $N_f-N_d$ is free of NMP zeros in this regime.

\addtolength{\textheight}{-3cm}   % This command serves to balance the column lengths

%%%%%%%%%%%%%%%%%%%%%%%%%%%%%%%%%%%%%%%%%%%%%%%%%%%%%%%%%%%%%%%%%%%%%%%%%%%%%%%%

   % \begin{figure}[thpb]
   %    \centering
   %    %\includegraphics[scale=1.0]{figurefile}
   %    \caption{Inductance of oscillation winding on amorphous
   %     magnetic core versus DC bias magnetic field}
   %    \label{figurelabel}
   % \end{figure}

%%%%%%%%%%%%%%%%%%%%%%%%%%%%%%%%%%%%%%%%%%%%%%%%%%%%%%%%%%%%%%%%%%%%%%%%%%%%%%%%
\section{CONCLUSIONS}\label{sec: conclusions}
In this paper, we proposed a novel IIG metric for robust fault detection 
inspired by the existing OOG metric.
A unified analysis was carried out on the performance limitations of the two metrics.
The future work will focus on the extension of the results to multiple-input-multiple-output systems.
Another research direction is to develop tractable methods for solving the optimal IIG.

\appendix

\begin{proof}[Proof of Theorem~\ref{thm: unified expression}]
We first prove the result for OOG.
Consider the original optimization problem~\eqref{eq: OOG} for OOG. 
Its Lagrange dual problem is derived as:
\begin{align}\label{eq: dual_OOG}
    \gamma^*_d = \min_{\gamma_d \in  \mathbb{R}_+} ~\gamma_d \quad \text{s.t.} ~\gamma_d \|y_r\|^2_{\mathcal{L}_2} \geq \|y_p\|^2_{\mathcal{L}_2},
\end{align}
where~$\gamma_d $ is the dual variable and $\gamma^*_d$ is the optimal value.
According to the property of the dual function, it yields an upper bound for the optimal value~$\gamma^*$ of~\eqref{eq: OOG}, i.e.,~$\gamma^*_d \geq \gamma^*$. 

Meanwhile, observe that~\eqref{eq: OOG} has a concave objective function and a convex constraint. Its optimal solution is achieved on the boundary of the constraint set, i,e,.~$\| y_r \|^2_{\mathcal{L}_2} = 1$.  
Therefore,~\eqref{eq: OOG} can be equivalently written as
$\gamma^* =  \sup_{a \in \mathcal{L}_{2e}, \|y_r\|^2_{\mathcal{L}_2} = 1} ~\frac{\|y_p\|^2_{\mathcal{L}_2}}{\|y_r\|^2_{\mathcal{L}_2}}$.  
This implies that~$\gamma^* \| y_r \|^2_{\mathcal{L}_2} \geq \| y_p \|^2_{\mathcal{L}_2}$ and~$\gamma^*$ belongs to the feasible solution set to~\eqref{eq: dual_OOG}. Consequently, 
$\gamma^\star \geq \gamma_d^\star$ since the latter is the optimal solution to~\eqref{eq: dual_OOG}. As a result, $\gamma_d^* = \gamma^\star$.

Based on Lemma~\ref{lem: FDI} and the right co-prime factorization in~\eqref{eq: RCPF}, a frequency domain condition for the inequality~$\gamma_d \|y_r\|^2_{\mathcal{L}_2} \geq \|y_p\|^2_{\mathcal{L}_2}$ in~\eqref{eq: dual_OOG} can be derived.
For~$s \notin \lambda_A, \text{Re}(s)=0$, we have
\begin{align*}
&\gamma_d \mathds{T}^{\top}_{a y_r}(\bar{s}) \mathds{T}_{a y_r}(s) \geq \mathds{T}^{\top}_{a y_p} (\bar{s})\mathds{T}_{a y_p}(s) \\
% \Leftrightarrow~ 
% &\gamma_d M_O^{-\top} N^{\top}_r  N_r M_O^{-1} \geq M_O^{-\top}N^{\top}_p  N_p M_O^{-1}\\ 
\Leftrightarrow~ 
& M_O^{-\top} (\gamma_d  N^{\top}_r  N_r - N^{\top}_p  N_p) M_O^{-1} \geq 0 \\
\Leftrightarrow~ &\gamma_d  N^{\top}_r  N_r - N^{\top}_p  N_p \geq 0 .
\end{align*}
Substituting the above inequality into~\eqref{eq: dual_OOG} yields 
\begin{align}\label{eq: FDP_OOG}
    &\gamma^*_{\omega}= \min_{\gamma_{d} \in \mathbb{R}_+} ~\gamma_{d} \\
    &\text{s.t.} 
    ~\gamma_d  N^{\top}_r(\bar{s})  N_r(s) - N^{\top}_p(\bar{s})  N_p(s) \geq 0,  s \notin \lambda_A, \textup{Re}(s) \geq 0, \notag
\end{align}
where~$\gamma^*_{\omega}$ is the optimal value of~\eqref{eq: FDP_OOG}.
Note that~\eqref{eq: FDP_OOG} is a generalized eigenvalue problem and~$\gamma^*_{\omega}$ essentially corresponds to the maximum generalized eigenvalue of the matrix pencil~$( N^{\top}_p(\bar{s})  N_p(s),N^{\top}_r(\bar{s})  N_r(s))$, i.e.,
\begin{align*}
    \gamma^*_{\omega} &= \sup_{s \notin \lambda_A, \textup{Re}(s) \geq 0} \lambda_{\textup{max}} \left( (N^{\top}_r(\bar{s})  N_r(s))^{-1} N^{\top}_p(\bar{s})  N_p(s) \right) \\ 
    &=\left\| \frac{N_p}{N_r} \right\|^2_{\mathcal{H}_\infty}.
\end{align*}

From Lemma~\ref{lem: FDI}, the frequency-domain inequality is a necessary condition for~$\gamma_d \|y_r\|^2_{\mathcal{L}_2} \geq \|y_p\|^2_{\mathcal{L}_2}$ with all possible trajectories. Therefore, the optimal value of~\eqref{eq: dual_OOG} is bounded from below by that of~\eqref{eq: FDP_OOG}, thus
$\gamma^* = \gamma^*_d \geq \gamma^*_{\omega}$.
Additionally, for system~\eqref{eq: LTI-SS} subject to periodic trajectories, the frequency-domain inequality is a necessary and sufficient condition. Therefore, the equalities~$\gamma^* = \gamma^*_d = \gamma^*_{\omega}$ hold for periodic trajectories. This completes the first part of the proof.

Subsequently, we prove the result for the reformulated IIG in~\eqref{eq: Reform_IIG}.
Since~$r = \mathds{T}_{dr}[d] + \mathds{T}_{fr}[f]$,~$\mathds{T}_{dr}[\tilde{d}] =  \mathds{T}_{fr}[f]$, and~$-2d=\tilde{d}$ with~$\xi=2$, we have $\tilde{d} = 2 \mathds{T}^{-1}_{dr} [r]$ and~$f = 2 \mathds{T}^{-1}_{fr} [r]$. The Lagrange dual problem of~\eqref{eq: Reform_IIG} is obtained as
\begin{align}\label{eq: dual_reform_IIG}
     \tilde{\gamma}^*_d = \min_{\tilde{\gamma}_d \in \mathbb{R}_+} ~\tilde{\gamma}_d \quad \text{s.t.} ~\tilde{\gamma}_d \|\mathds{T}^{-1}_{dr} [r]\|^2_{\mathcal{L}_2} \geq 4\|\mathds{T}^{-1}_{fr} [r]\|^2_{\mathcal{L}_2},
\end{align}
where~$\tilde{\gamma}_d$ is the dual variable.
The rest of this part follows a similar procedure as the first part. 

First, it can be shown that the optimal value of~\eqref{eq: dual_reform_IIG} is equal to that of~\eqref{eq: Reform_IIG}, i.e.,~$\tilde{\gamma}^*_d = \tilde{\gamma}^*$.
Then, using the left co-prime factorization in~\eqref{eq: LCPF} and Lemma~\ref{lem: FDI}, we obtain the frequency-domain inequality for~$\tilde{\gamma}_d\|\mathds{T}^{-1}_{dr} [r]\|^2_{\mathcal{L}_2} \geq 4\|\mathds{T}^{-1}_{fr} [r]\|^2_{\mathcal{L}_2}$, i.e.,
\begin{align*}
    &\tilde{\gamma}_d \mathds{T}^{-\top}_{d r}(\bar{s}) \mathds{T}^{-1}_{dr}(s) \geq 4 \mathds{T}^{-\top}_{fr}(\bar{s}) \mathds{T}^{-1}_{fr}(s),% \\
    % \Leftrightarrow~
    % &\tilde{\gamma}_d M^{\top}_{I} N^{-\top}_d N^{-1}_d M_I   \geq 4 M^{\top}_{I}N^{-\top}_f N^{-1}_f M_I \\
    % \Leftrightarrow~ 
    % &M^{\top}_{I} (\tilde{\gamma}_d N^{-\top}_d N^{-1}_d - 4 N^{-\top}_f N^{-1}_f)M^{\top}_I \geq 0 \\
    %\Leftrightarrow~ 
    %&\tilde{\gamma}_d N^{-\top}_d N^{-1}_d - 4 N^{-\top}_f N^{-1}_f \geq 0.
\end{align*}
which is equivalent to $\tilde{\gamma}_d N^{-\top}_d N^{-1}_d - 4 N^{-\top}_f N^{-1}_f \geq 0$. Substituting the frequency-domain inequality into~\eqref{eq: dual_reform_IIG} yields
\begin{align*}
    &\tilde{\gamma}^*_{\omega}= \min_{\tilde{\gamma}_d \in \mathbb{R}_+} ~\tilde{\gamma}_{d} \notag\\
    &\text{s.t.} 
    ~\tilde{\gamma}_d N^{-\top}_d N^{-1}_d - 4 N^{-\top}_f N^{-1}_f \geq 0,  s \notin \lambda_A, \textup{Re}(s) \geq 0.
\end{align*}
The optimal value~$\tilde{\gamma}^*_{\omega}$ is the maximum generalized eigenvalue of the matrix pencil~$(4 N^{-\top}_f(\bar{s}) N^{-1}_f(s), N^{-\top}_d(\bar{s}) N^{-1}_d(s) )$ and equals
$\tilde{\gamma}^*_{\omega} = 4 \left\| {N_d}/{N_f} \right\|^2_{\mathcal{H}_\infty}$.
For systems of all trajectories,~$\tilde{\gamma}^*_{\omega}$ is a lower bound for~$\tilde{\gamma}^*$, i.e.,~$\tilde{\gamma}^* \geq \tilde{\gamma}^*_{\omega}$.
For systems subject to periodic trajectories, the equality is achieved, as the frequency-domain inequality becomes a necessary and sufficient condition.
This completes the proof.
\end{proof}

%The following lemma is introduced for proof of Theorem~\ref{thm: performance}.

\begin{lemma}[Poisson integral formula~{\cite[Corollary~A.6.3]{seron2012fundamental}}]
    Let $f(s)$ be analytical and of bounded magnitude in the right half of the complex plane. Consider a point $s_0=\sigma_0+j\omega_0$ with $\textup{Re}(s_0)>0$. It holds that
    \begin{align}\label{eq: Poisson}
        \log |f(s_0)| = \frac{1}{\pi} \int_{-\infty}^{\infty} \log |f(j \omega)| \frac{\sigma_0}{\sigma_0^2+(\omega_0-\omega)^2} ~\textup{d} \omega.
    \end{align}
\end{lemma}

\begin{proof}[Proof of Theorem~\ref{thm: performance}]
    We only demonstrate the validity of the performance limitation for~$||S_O||_{\mathcal{H}_{\infty}}$, as the derivation process for the limitation of~$||S_I||_{\mathcal{H}_{\infty}}$ follows a similar path.
    Based on~\eqref{eq: Poisson}, setting $f(s) = \tilde{S}_O(s)$ and noting that $S_O(\beta_k) = \tilde{S}_O(\beta_k) \mathcal{B}_{S_O}(\beta_k)=1$ for each $\beta_k \in \mathcal{Z}_{P_O}$, $\beta_k = \sigma_{\beta_k} + j \omega_{\beta_k}$, we have
    \begin{align*}
        &\log |\tilde{S}_O(\beta_k)| 
       =  \log |\mathcal{B}^{-1}_{S_O} (\beta_k)| \\
       =~&\frac{1}{\pi} \int_{-\infty}^{\infty} \log |S_O(j \omega)| \frac{\sigma_{\beta_k}}{\sigma_{\beta_k}^2+(\omega_{\beta_k}-\omega)^2} \, \text{d} \omega,
    \end{align*}
    where $|\mathcal{B}_{S_O}(j\omega)|=1$ (therefore $|S_O(j \omega)| = |\tilde{S}_O(j\omega)|$) is utilized in the second equality.
    Since the integration $\int_{-\infty}^{\infty}  \frac{\sigma_0}{\sigma_0^2+(\omega_0-\omega)^2} ~\text{d} \omega = \pi$, it holds that
    \begin{align}\label{eq: SO ineq1}
        ||S_O||_{\mathcal{H}_{\infty}} \geq |\mathcal{B}^{-1}_{S_O} (\beta_k)| = \left\vert \prod^{n_\alpha}_{i=1} \frac{\beta_k+\bar{\alpha}_i}{\beta_k-\alpha_i} \right\vert.
    \end{align}

    Subsequently, setting $f(s) = \tilde{P}_O(s)$ in~\eqref{eq: Poisson} and noting that $P_O(\alpha_h) = \tilde{P}_O(\alpha_h) \mathcal{B}_{P_O}(\alpha_h)=1$ for each $\alpha_h \in \mathcal{Z}_{S_O}$, $\alpha_h = \sigma_{\alpha_h} + j \omega_{\alpha_h}$, we have
    \begin{align*}
        &\log |\tilde{P}_O(\alpha_h)| = 
        \log |\mathcal{B}^{-1}_{P_O} (\alpha_h)| =\\
        &\frac{1}{\pi} \int_{-\infty}^{\infty} \log |P_O(j \omega)| \frac{\sigma_{\alpha_h}}{\sigma_{\alpha_h}^2+(\omega_{\alpha_h}-\omega)^2} \, \text{d} \omega.
    \end{align*}
    Similar to the above analysis, we have  
        $||P_O||_{\mathcal{H}_{\infty}} \geq |\mathcal{B}^{-1}_{P_O} (\alpha_h)| = \left\vert \prod^{n_\beta}_{i=1} \frac{\alpha_h+\bar{\beta}_i}{\alpha_h-\beta_i} \right\vert$.
    Moreover, since $P_O = 1 - S_O$, it holds that $1 + ||S_O||_{\mathcal{H}_{\infty}} \geq  ||P_O||_{\mathcal{H}_{\infty}}$.
    Therefore, we obtain
    \begin{align}\label{eq: SO ineq2}
        ||S_O||_{\mathcal{H}_{\infty}} \geq  |\mathcal{B}^{-1}_{P_O} (\alpha_h)| -1.
    \end{align}
    Since~\eqref{eq: SO ineq1} holds for all~$\beta_k$ and~\eqref{eq: SO ineq2} holds for all~$\alpha_h$, the performance limitation of $S_O$ satisfies \eqref{eq: performance bound}.
    % \begin{align*}
    %     ||S_O||_{\mathcal{H}_{\infty}} \geq \max_{\alpha_h\in\mathcal{Z}_{S_O},\beta_k\in\mathcal{Z}_{P_O}} \left\{ |\mathcal{B}^{-1}_{S_O}(\beta_k)|, |\mathcal{B}^{-1}_{P_O}(\alpha_h)|-1 \right\}.
    % \end{align*}
    This completes the proof.
\end{proof}

%%%%%%%%%%%%%%%%%%%%%%%%%%%%%%%%%%%%%%%%%%%%%%%%%%%%%%%%%%%%%%%%%%%%%%%%%%%%%%%%

% \addtolength{\textheight}{-0cm}   % This command serves to balance the column lengths
                            
\balance                                
\end{document}